\documentclass{amsart}
\usepackage{amsmath,amsthm,amsfonts, amssymb,amscd}
\usepackage[all]{xy}

\newtheorem{theorem}{Theorem}[section]
\newtheorem{lemma}[theorem]{Lemma}
\newtheorem{example}[theorem]{Example}
\newtheorem{remark}[theorem]{Remark}
\newtheorem{proposition}[theorem]{Proposition}

\newcommand{\tr}{\mbox{\rm tr}}
\begin{document}
\title[Regular Gleason Measures and Generalized Effect Algebras]{Regular Gleason Measures and Generalized Effect Algebras}
\author[Anatolij Dvure\v{c}enskij, Ji\v{r}\'i Janda]{Anatolij Dvure\v{c}enskij$^{1,2}$, Ji\v{r}\'i Janda$^3$}
\date{}%
\maketitle
\begin{center}  \footnote{


The authors acknowledge the support (A.D.) by the Slovak Research and Development Agency under the contract APVV-0178-11, the grant VEGA No. 2/0059/12 SAV, ESF
Project CZ.1.07/2.3.00/20.0051, (J.J.) ESF
Project CZ.1.07/2.3.00/20.0051 and Masaryk University grant 0964/2009  }
Mathematical Institute,  Slovak Academy of Sciences,\\
\v Stef\'anikova 49, SK-814 73 Bratislava, Slovakia\\
$^2$ Depart. Algebra  Geom.,  Palack\'{y} Univer.\\
17. listopadu 12,
CZ-771 46 Olomouc, Czech Republic\\
$^3$ 	Depart. Math. and Statistics, Faculty of Science\\
Masaryk University, Kotl\'a\v{r}sk\'a 267/2, CZ-611 37  Brno, Czech Republic\\
E-mail: {\tt
dvurecen@mat.savba.sk} \qquad {\tt xjanda@math.muni.cz}
\end{center}

\begin{abstract}
We study measures, finitely additive measures, regular measures, and $\sigma$-additive measures that can attain even infinite values on the quantum logic of a Hilbert space. We show when particular classes of non-negative measures can be studied in the frame of generalized effect algebras.
\end{abstract}

\section{Introduction}

A basic theoretical tool of quantum mechanical measurements is based on the famous Gleason Theorem \cite{Gle} which says that if $H$ is a separable Hilbert space over real, complex  or quaternion numbers, $\dim H\ne 2,$ then there is a one-to-one correspondence between $\sigma$-additive states, $m$, on the system $\mathcal L(H)$ of all closed subspaces of $H$, and the set of von Neumann operators on $H$, i.e. positive Hermitian operators, $T$, of trace equal to $1$, given by
$$
m(M)=\mathrm{tr}(TP_M),\quad M\in \mathcal L(H), \eqno(1.1)
$$
where $P_M$ is an orthogonal projector onto $M$. Formula (1.1) holds also for any completely additive state for every $\mathcal L(H)$, $\dim H \ne 2$.  For an extension of the Gleason Theorem to non-separable Hilbert spaces, see e.g. \cite{Dvu2, DvuG}. In such a cases, if $\dim H\ne 2$ is a non-measurable cardinal, then every $\sigma$-additive state is completely additive. In addition, formula (1.1) is valid also for bounded signed $\sigma$-additive measures, \cite{DvuG}.

However, there are also important measures attaining even infinite values, like, $\dim M$, $M \in \mathcal L(H)$, or attaining both positive and negative values, e.g. $\mathrm{tr}(T(AP_M+P_MA)/2)$, $M \in \mathcal L(H)$, where $A$ is a Hermitian operator on $H$ that is neither positive nor negative. Also for them it is possible to find an extension of the Gleason formula (1.1), however, in such a case, they are induced  rather by bilinear forms than by trace operators, see e.g. \cite{LuSh} as it will be shown by formula (3.1). Therefore, we  call Gleason  measures all measures that can be expressed by a generalized Gleason formula.

If we study finitely measures, $m$, with improper values, then regularity of $m$, i.e. the property ``given $M \in \mathcal L(H)$, the value $m(M)$ can be approximated by values of $m(N)$, where $N$ is a finite-dimensional subspace of $M$", does not implies that the measure is $\sigma$-additive as in the case that $m$ is a finite measure. By \cite[Lem 3.7.8. Thm 3.7.9]{DvuG}, every positive regular measure $m$ that is $\sigma$-finite (i.e. $H$ can be split into a sequence of mutually orthogonal subspaces of finite measure)  is in a  one-to-one correspondence with densely defined positive bilinear forms. Consequently, regular measures are Gleason measures, too. Since, not every bilinear form is induced by a Hermitian operator, bilinear forms define more regular measures than do bilinear forms induced by Hermitian trace operators via (1.1). For more details on measures on $\mathcal L(H)$, we recommend to consult the book \cite{DvuG}.

Foulis and Bennett \cite{FoBe} have introduced effect algebras that are inspired by mathematical models of quantum mechanics. They model both a two-valued reasoning and a many valued one. They have a least and a top element, $0$ and $1$, respectively, and a basic operation is a partially defined addition $+$. For example, POV-measures have within these models a natural form as observables.
More general structures than effect algebras are generalized effect algebras \cite{DvPu}, which have a similar partially defined addition, but the top element $1$ is not guaranteed. For additional information on effect algebras and generalized effect algebras, we recommend the monograph \cite{DvuG}.

The aim of this paper is to study some properties of measures using methods of generalized effect algebras. We describe a structure of generalized effect algebras of finitely additive Gleason measures, regular measures, and $\sigma$-measures. In addition, we study also questions connected with a Dedekind monotone $\sigma$-convergence.

The paper is organized as follows. Section 2 presents elements of bilinear forms on a Hilbert space and frame functions that are semibounded. Regular measures and their properties are studied in Section 3 together with different kinds of generalized effect algebras. In Section 4, we deal with monotone convergence properties of $\sigma$-additive measures.

\section{Bilinear Forms, Signed Measures, and Frame Functions}

Let $S$ be a real or complex inner product space with an inner product $(\cdot,\cdot)$ which is linear in the left argument and antilinear in the right one. Sometimes we call $S$ a pre-Hilbert space. If $S$ is complete with respect to the norm $\|\cdot\|$ induced by the inner product, $S$ is said to be a {\it Hilbert space}. Let $H$ be a Hilbert space and let $\mathcal L(H)$ be the system of all closed subspaces of $H$.  We denote by $\mathbf 0$ the zero vector of $H$ and if $x_1,\ldots,x_n \in H,$ ${\rm sp}(x_1,\ldots,x_n)$ denotes the span generated by $x_1,\ldots,x_n$. We denote by $\mathcal S(H)$ the unit sphere of $H$.  Similarly, if $S$ is a linear submanifold of $H,$ $\mathcal S(S)$ denotes the unit sphere of $S$.

Any system of mutually orthogonal unit vectors $\{x_i\}$ of an inner product space  $S$ is denoted as ONS. An ONS $\{x_i\}$ is said to be (i) a {\it maximal} ONS (MONS for short) if $x$ is a vector orthogonal to each vector $x_i,$ then $x = \mathbf 0;$ (ii) {\it orthonormal base} (ONB for short), if for any vector $x \in S,$ we have $x = \sum_i (x,x_i) x_i$.

A Hermitian operator $T$ on $H$ is said to be a {\it trace operator} (or of {\it finite trace}) if there is a real constant $\mathrm{tr}(T)$, called the {\it trace} of $T$, such that for any orthonormal basis $\{x_i\}$, $\mathrm{tr}(T)=\sum_i (Tx_i,x_i)$. We denote by $\mathrm{Tr}(H)$ the class of Hermitian operators with finite trace.

Now we introduce bilinear forms. Let $D(t)$ be a submanifold of a complex Hilbert space $H$ and let $D(t)$ be dense in $H$.  A mapping $t : D(t) \times D(t) \rightarrow \mathbb{C}$ is called a {\it bilinear form} (or a {\it sesquilinear form}) if and only if it is
additive in both arguments and $(\alpha x, \beta y) = \alpha{\overline \beta} (x,y)$ for all $\alpha, \beta \in \mathbb{C}$, $x,y \in D(t)$, where $\overline{\beta}$ is the complex conjugation of $\beta.$  A complex-valued function $\hat t$ with the definition domain $D(t)$ is defined by $\hat t(x)=t(x,x),$ $x \in D(t)$, and it is said to be a {\it quadratic form} induced by the bilinear form $t,$ \cite[p. 12]{Hal}.
A bilinear form $t$ is called {\it symmetric} if $t(x,y) = \overline{t(y,x)}$.

A symmetric bilinear form $t$ is (i) {\it semibounded} if $\inf\{t(x,x): x \in \mathcal S(D(t))\}>-\infty$,  (ii) {\it positive}
if $t(x,x)\ge 0$, $x \in \mathcal S(D(t))$, and (iii) {\it bounded} if $\sup\{|t(x,x)|: x \in \mathcal S(D(t))\}<\infty$; otherwise $t$ is called {\it unbounded}. We denote by $\mathcal{BF}(H)$ and $\mathcal{PBF}(H)$ the set of all bilinear forms and positive bilinear forms, respectively, on $H$.
We denote by $o$ the bilinear with $D(o)= H$ defined by $o(x,y)=0$ for all $x,y \in H$; then $o \in \mathcal{PBF}(H).$

On  $\mathcal{BF}(H),$ we can define a {\it usual sum} $t+s$ for each $t, s$ on $ D(t+s):=D(t) \cap D(s)$
by $(t+ s)(x,y):=t(x,y)+ s(x,y)$ for all  $x,y \in D(t) \cap D(s)$, and the multiplication by a scalar $\alpha\in \mathbb C$ by $(\alpha t)(x,y):= \alpha t(x,y)$ for $x,y \in D(\alpha t):=D(t).$

Given two densely defined positive bilinear forms $t$ and $s,$ we write
$$
t \preceq s \eqno(2.0)\label{prec}
$$
if and only if $D(t)\supseteq D(s)$ and $t(x,x)\le s(x,x)$ for all $x \in D(s).$ Then
$\preceq$ is a partial order on $\mathcal{PBF(H)},$ and the bilinear form $o$ is the least element of the set $\mathcal {PBF(H)},$ i.e. $o \preceq t$ for any $t \in \mathcal{PBF(H)}.$

A bilinear form $s$ is an {\it extension} of a bilinear form $t$ iff $D(t) \subseteq D(s)$ and, for every $x,y \in D(t)$, $t(x,y) = s(x,y);$ we will write $s_{\mid D(t)}=t.$ A bilinear form $t$ is {\it closable} if it has some closed extension. Then the least closed extension of a closable $t$ is said to be a {\it closure} and we denote it by $\overline{t}$.

For any densely defined operator $A $ with domain $D(A)$ on $H$, the mapping $t(x,y):=(Ax, y)$ is a bilinear form on $D(A)$. We say that a bilinear form $t$
{\it corresponds} to an operator $A$  iff $D(A) \subseteq D(t)$ and, for every $x, y \in D(A),$ $t(x,y)=(Ax,y)$. We say that $t$ is {\it generated} by $A$ if $t$ corresponds to $A$ and $D(A)=D(t)$. We note that there are bilinear forms that are not generated by any operator. In addition, there is a one-to-one correspondence between closed positive symmetric bilinear forms, $t$, and positive self-adjoint operators, $A$, given by $D(t)=D(A^{1/2})$, and $t(x,y)=(A^{1/2}x,A^{1/2}y)$, $x,y \in D(t)$.

Given a symmetric semibounded bilinear form $t,$ we can equip its domain $D(t)$ with
an inner product $(x, y)_t := t(x, y) + (1+m_t)(x, y)$. In this way $D(t)$ becomes a pre-Hilbert space. Whenever $D(t)$ with $(x, y)_t$ is a Hilbert space, we call $t$ {\it closed.}

A bilinear form $s$ is an {\it extension} of a bilinear form $t$ iff $D(t) \subseteq D(s)$ and, for every $x,y \in D(t)$, $t(x,y) = s(x,y);$ we will write $s_{\mid D(t)}=t.$ A bilinear form $t$ is {\it closable} if it has some closed extension.

An important result which is also a matter of interest for us was proved in  \cite[Thm 2.1, Thm 2.2]{Sim}:

\begin{theorem}\label{bs}
Let $t$ be a densely defined positive symmetric bilinear form on a Hilbert space $H$. Then there exist two positive symmetric bilinear forms
$t_r$ and $t_s$ such that $D(t)=D(t_r)=D(t_s)$ such that
$$
t = t_r + t_s,
\label{eq:rs}
$$
where $t_r$ is the largest closable bilinear form less than $t$ in the ordering $\preceq.$
\end{theorem}

In view of  Theorem \ref{bs}, the components $t_r$ and $t_s$ of the positive symmetric bilinear form $t$, are said to be the {\it regular} and {\it singular} part of $t,$ respectively.

If $t_{s}$ is identical zero, the bilinear form  $t$ is said to be {\it regular},    and if $t_{r}$ is identical zero, $t$ is  said to be {\it singular}. Hence, the bilinear form $o$ is the least positive symmetric bilinear form and is a unique positive symmetric bilinear form that is simultaneously regular and singular.

The following example presents a nonzero everywhere  defined positive singular symmetric bilinear form.

\begin{example}\label{ex:sing}
Let $\{e_n\}$ be an orthonormal basis of an infinite-dimensional
Hilbert space $H$ which is a part of a Hamel basis $\{e_i\}_{i \in I}$ of $H$
consisting of unit vectors from $H$. Fix an element $e_{i_0}$, $i_0 \in I,$
which does not belong to the orthonormal basis $\{e_n\}$, and define a linear operator $T$ on $ H$ by

$$
T(\sum_i \alpha_i \,e_i) = \alpha_{i_0} \, e_{i_0},
$$
where $\alpha_{i_0}$ is the scalar corresponding to $e_{i_o}$ in the decomposition $x = \sum_i \alpha_i \,e_i$, $x \in  H$, with  respect to the given Hamel basis. Then $T$ is an everywhere defined unbounded linear operator.

If we define a bilinear form $t$ with $D(t) =  H$ via
$$
t(x,y) = (Tx,Ty),\ x,y \in H,
$$
then $t$ is a nonzero everywhere  defined positive singular symmetric bilinear form.
\end{example}

A mapping $m:\mathcal L(H) \to [-\infty,\infty]$ is said to be a {\it finitely additive signed measure} if (i) $m({\rm sp}(\mathbf 0))=0, $ and (ii) $m(M \vee N)=m(M)+m(N)$ whenever $M \perp N,$ $M,N \in \mathcal L(H)$.  We note that $m$ can attain from the improper values $\{-\infty,\infty\}$ at most one value.

If for $m:\mathcal L(H)\to [-\infty,\infty]$ with the above (i) and instead of (ii) we have (ii)' $m(\bigvee_{i\in I}M_i)=\sum_{i\in I}m(M_i)$ whenever $M_i \perp M_j$ for $i\ne j$, $i,j\in I$ holds for any index set $I,$  $m$ is said to be a {\it completely additive signed measure} (also totally additive); if it holds only for any countable index set $I$, $m$ is said to be a $\sigma$-{\it additive signed measure}. If $m(M)\ge 0$ for any $M\in \mathcal L(H),$ we are speaking on a {\it finitely additive measure}, a {\it completely additive measure}, and a $\sigma$-{\it additive measure}, respectively. Every finitely additive measure is monotone.  If, in addition, $m(M)\ge 0$ for any $M \in \mathcal L(H)$ and  $m(H)=1,$ $m$ is said to be a {\it finitely additive state}, a {\it completely additive measure} and a $\sigma$-{\it additive state}, respectively. As we have already mentioned in Introduction, if $m$ is a completely additive state on $\mathcal L(H)$, $\dim H\ne 2,$ then by the Gleason Theorem, see e.g. \cite[Thm 3.2.21]{DvuG}, there is a unique positive Hermitian operator $T$ on $H$ with trace $\mbox{tr}(T)=1$ such that
$$
m(M)=\mbox{tr}(TP_M),\quad M \in \mathcal L(H). \eqno(2.1)
$$
If  $H$ is a finite-dimensional Hilbert space with $\dim H\ne 2$ and $m$ is a finite finitely additive signed measure such that $\inf \{m(\mbox{sp}(x)): x \in \mathcal S(H)\}> -\infty,$ then $m$ is bounded and there is a unique Hermitian operator $T$ on $H$ such that (2.1) holds, see \cite{She, Dvu1} or \cite[Thm 3.2.16]{DvuG}.

A finitely additive measure $m$ is $\sigma$-{\it finite} if there is a sequence $\{M_i\}$ of mutually orthogonal subspaces of $H$ such that $\bigvee_i M_i = H$ and every  $m(M_i)$ is finite.

A function $f:\mathcal S(H) \to [-\infty, \infty]$ is said to be
a {\it frame function} if, for any $M \in \mathcal L(H),$ there is a constant $W_M \in [-\infty,\infty]$ such that $\sum_i f(x_i)=W_M$ holds for any ONB $\{x_i\}$ in $M$.  The constant $W_M$ is said to be the {\it weight} of $f$ on $M \in \mathcal L(H)$.  It is worthy of recalling that (i) if $W_H=\infty,$ then $W_M >-\infty$ for any $M\in \mathcal L(H),$ (ii) if $W_H=-\infty,$ then $W_M <\infty$ for any $M \in \mathcal L(H)$.  In what follows, we will study frame functions $f$ such that $f(x)> -\infty$ for any $x \in \mathcal S(H)$.

We note that if $f:\mathcal S(H)\to \mathbb R$ has the property that there is a constant $W\in \mathbb R$ such that $\sum_if(x_i)=W$ for any ONB $\{x_i\},$ then $f$ is a frame function. Indeed, if $M\in \mathcal L(H)$ is given, take two ONB's $\{x_i\}$ and $\{y_i\}$ of $M$ and let $\{z_j\}$ be an ONB in $M^\bot$.  Then $\{x_i\}\cup \{z_j\}$ and $\{y_i\}\cup \{z_j\}$ are ONBs of $H$.  Hence, $\sum_i f(x_i)+\sum_j f(z_j)=W= \sum_i f(y_j)+ \sum_j f(z_j)$ and the series converge absolutely, so that $\sum_i f(x_i)= \sum_i f(y_j)$.

Let $m$ be a completely additive signed measure. Then the function $f(x) := m({\rm sp}(x)),$ $x \in \mathcal S(H),$ is a frame function. Conversely, any frame function $f$ defines a completely additive measure on $\mathcal L(H)$.  Indeed, let $\{x_i\}$ be an ONB of $M \in \mathcal L(H)$.   Then the mapping $m(M):=\sum_i f(x_i)=W_M,$ $M \in \mathcal L(H),$ where $W_M$ is the weight of $f$ on $M$ is a completely additive measure.

Let $S$ be a linear submanifold of $H$ dense in $H$.  A function $f:\mathcal S(S) \to [-\infty,\infty]$ is said to be a {\it frame type function} on $H$ if (i) for any ONS $\{x_i\}$ in $S$, either $\{f(x_i)\}$ is summable or $\sum_i f(x_i) = \infty$,  
and (ii) for any finite-dimensional subspace $K$ of $S,$ $f|\mathcal S(K)$ is a frame function on $K$.
To show the relationship between a frame type function $f: \mathcal S(S)\to [-\infty,\infty]$ and a dense submanifold $S$, we will write $f \sim (f,S)$.
If we change the latter (i) to (i)' for any ONS $\{x_i\}$ in $S$, either $\{f(x_i)\}$ is summable or $\sum_i f(x_i) = -\infty$, then $-f$ satisfies (i), and vice-versa. 

We note that there is a result from theory of series, \cite[Thm 388, p. 319]{Fic}, Riemann's Rearrangement Theorem, saying that if $\{a_n\}$ is an infinite sequence of real numbers such that $\sum_n a_n$ is finite but $\sum_n |a_n|=\infty$, then for any $a \in [-\infty,\infty]$, there is a rearrangement $\{a_{n'}\}$ of $\{a_n\}$ such that $a=\sum_{n'}a_{n'}$.

Therefore, let $f$ be a frame type function and  $\{x_i\}$ be an ONS in $S$. By the Riemann Rearrangement Theorem, if $\{f(x_i)\}$ is summable, then $\sum_i|f(x_i)|<\infty$, otherwise we can rearrange $\{x_i\}$ to $\{x_{i'}\}$ such that $\sum_{i'}f(x_i)=-\infty$ which is absurd.

It is important to emphasize that frame type functions were originally introduced in \cite{DoSh} (cf. \cite[Sect 3.2.4]{DvuG}) in a stronger form than our notion, because in \cite{DoSh} it was supposed that (i) has the form ``$\{f(x_i)\}$ is summable for any ONS $\{x_i\}$ in $S$." We say that a frame type function $f \sim (f,S)$  is {\it finite} on $\mathcal S(S)$ if, for every $x \in \mathcal S(S)$, we have $-\infty < f(x) < \infty$. 

We note that if $f\sim (f,S)$ is a finite frame type function, $x \in \mathcal S(S)$ and $|\lambda|=1,$ then $f(x)=f(\lambda x)$.

It is clear that any frame function on $H$ such that $W_{H} > - \infty$ is a frame type function on $H$.
In addition, if $f$ is a frame type
function on $H$ with given $S$, then  if $K$ is a closed subspace
of $H$ such that $S \cap K$ is dense in $K$, then $f_K := f|{\mathcal
S}(S\cap K)$ is a frame type function on $K$.

\begin{example}\label{ex:2.1}
Let $H$ be a complex separable Hilbert space, $\dim H=\aleph_0$.
\begin{enumerate}

\item[{\rm (i)}] If $T$ is a Hermitian trace operator, then $f(x)=(Tx,x),$ $x \in \mathcal S(H),$ is a frame function on $H$.

\item[{\rm (ii)}] If $T_1$ is a positive unbounded operator and $T$ is a positive Hermitian trace operator, then $f(x)=(T_1^{1/2}x,T_1^{1/2}x)-(Tx,x),$ $x \in D(T^{1/2})$, is a finite frame type function with $S=D(T_1^{1/2}),$ where $D(T_1^{1/2})$ is the domain of definition of $T_1^{1/2}$.

\item[{\rm (iii)}]  Let $M$ be a fixed closed subspace of $H,$  $M \subset H,$ $M \ne H,$ let $T_M$ be a trace operator on $M$. Then
    \[   f(x) = \left\{ \begin{array}{ll}
(T_Mx,x) &\quad {\rm if }\
x \in M,\\
 \infty & \quad  {\rm otherwise},
\end{array}
\right.
\]
is an unbounded frame type function with $S=H$.

\item[{\rm (iv)}] Let $\{e_i\}$ be an ONB and let us define a diagonal operator $T$ on $H$ such that $Te_i=(-1)^i e_i$. Then $T$ is a Hermitian operator defined everywhere on $H$ and it defines a bounded symmetric  bilinear form $f(x)=(Tx,x)$, $x \in \mathcal S(H)$, with $S=H$, that is not a frame type function.
\end{enumerate}

\end{example}

The case (iv) can be generalized as follows. Let $\{a_n\}$ be a bounded sequence of real numbers. We set $a_n^+=\max\{a_n,0\}$ and $a^-=-\inf \{a_n,0\}$. Let $\{e_n\}$ be a fixed ONB and we define three bounded Hermitian operators $T$, $T^+$ and $T^-$ on $H$ by $Te_n =a_ne_n$, $T^+e_n =a_n^+e_n$, and  $T^-e_n =a_n^-e_n$. Then $T=T^-+T^-$, and $T^+$ and $T^-$ are positive and negative parts of $T$, respectively. Using the Riemann Rearrangement Theorem, we have the following four cases:

\begin{enumerate}
\item[(I)] $\sum_n a_n$ is convergent, $\sum_n |a_n| <\infty$.
\item[(II)] $\sum_n a_n$ is convergent, $\sum_n |a_n|=\infty$.
\item[(III)] $\sum_n a_n=\infty$, and $\sum_n a_n^+=\infty$, $\sum_n a_n^-<\infty$ or $\sum_n a_n^+<\infty$, $\sum_n a_n^-=\infty$.
\item[(IV)] $\sum_n a_n^+=\infty= \sum_n a_n^-$.
\end{enumerate}

If we set $f(x)=(Tx,x)$, $x \in \mathcal S(H)$, then in case (I), $f$ is a finite bounded frame function, in cases (II) and (IV) it is no frame type function.

Case (III): Let $s_n=a_1+\cdots+a_n$, $n\ge 1$. Then $s_n=s_n^+-s^-_n$, where $s_n^+=a_1^+ +\cdots + a_n^+$ and $s_n^-=a_1^- +\cdots + a_n^-$, $n \ge 1$.
Then either $\{s_n^+\}$ is convergent and $\{s_n^-\}$ is divergent or vice-versa. Let $\{x_i\}$ be an arbitrary ONB of $H$. The first possibility implies $T^+\in \mathrm{Tr}(H)$ and $T^-\notin \mathrm{Tr}(H)$. Then $-(T^-x_i,x_i)\le (Tx_i,x_i)\le (T^+x_i,x_i)$, and this gives $T^+ \notin \mathrm{Tr}(H)$ which is absurd. Thus we have only the second possibility, i.e.,
$T^+\notin \mathrm{Tr}(H)$ and $T^-\in \mathrm{Tr}(H)$.

Let $K>0$ be given. Since $\sum_i(T^+x_i,x_i)=\infty,$ there is an integer $n_0>0$ such that, for every integer $n\ge n_0$, we have
$\sum_{i=1}^n (T^+x_i,x_i)> K+\mathrm{tr}(T^-)$ so that $\sum_{i=1}^n(Tx_i,x_i)=\sum_{i=1}^n(T^+x_i,x_i)- \sum_{i=1}^n(T^-x_i,x_i)> K+\mathrm{tr}(T)- \sum_{i=1}^n(T^-x_i,x_i)>K$ which yields $\sum_i(Tx_i,x_i)=\infty$.

Now let $\{z_j\}$ be an arbitrary ONS in $H$. We have two cases (i) $\sum_j (T^+z_j,z_j)=\infty$ and (ii)  $\sum_j (T^+z_j,z_j)<\infty$. In the case (i), using the same way as that for $\{x_i\}$, we have $ \sum_j (Tz_j,z_j)=\infty$, and in the case (ii), we have $\sum_j (Tz_j,z_j)$ is convergent. Consequently, we have proved the following statement:

\begin{proposition}\label{pr:diag}
Let $\{a_n\}$ be a bounded sequence of real numbers, $\{e_n\}$ be a fixed ONB on a Hilbert space $H$, $\dim H=\aleph_0$, and let $Te_n=a_ne_n$, $n\ge 1$. Then $f(x)=(Tx,x)$, $x \in \mathcal S(H)$, defines a frame type function if and only if either $\sum_n|a_n|<\infty$ or $\sum_n a_n^+<\infty$, $\sum_n a_n^-=\infty$. The first case implies $f$ is a bounded frame function and the second one implies $f$ is a bounded frame type function.
\end{proposition}

It is important to note that if $\dim H>1$ is finite, there are unbounded frame functions on $H$, \cite{She}, \cite[Prop 3.2.4]{DvuG}. If $H$ is infinite-dimensional, then the surprising result \cite{DoSh} asserts that any finite frame type function $f$ such that $\{f(x_i)\}$ is summable for any ONS $\{x_i\}$ in $S$ is necessarily  bounded, see also \cite[Thm 3.2.20]{DvuG}.


A frame type function $f$ on $H$ defined on a dense linear submanifold $S$ of $H$ is said to be (i) {\it semibounded} if $\inf \{f(x)\colon x \in \mathcal S(S)\} >-\infty$, and (ii) {\it bounded} if $\sup\{|f(x)| \colon x \in \mathcal S(S)\}< \infty.$

Semibounded frame type functions are connected with semibounded bilinear forms as it follows from the following theorem.

\begin{theorem}\label{th:2.2}
Let $f \sim (f,S)$ be a semibounded finite frame type function on $H,$ $\dim H \ne 2$.  There is a unique semibounded symmetric bilinear form $t$ with $D(t)=S$ and $f(x)=t(x,x),$ $x \in \mathcal S(S)$.
\end{theorem}

\begin{proof}
First, let $H$ be finite-dimensional and $\dim H \ne 2$.  Then $H=S$ and $f$ is a frame function. Since $f$ is semibounded, then $f$ is bounded on $\mathcal S(H)$.  Indeed, let $y$ be any unit vector of $H$ and complete it by unit vectors $x_1,\ldots,x_n$ to be $\{y,x_1,\ldots,x_n\}$ an ONB of $H$.  Let $K =\inf \{f(x) \colon x \in \mathcal S(H)\}$.  Then $|f(y)| \le W_H +n|K|$.  By Gleason's Theorem \cite[Thm 3.2.1]{DvuG}, there is a unique Hermitian operator $T$ on $H$ such that $f(x)=(Tx,x),$ $x \in \mathcal S(H),$ and if we set $t(x,x) =(Tx,x),$ $x \in H,$ then $t$ is a bounded symmetric bilinear form such with $D(t)=S=H$ and $f(x)=t(x,x),$ $x \in \mathcal S(H)$.

Now let $\dim H=\infty$.  Let $N$ be a finite-dimensional subspace of $S,$ $\dim N \ge 3$.  As in the first part of the present proof, $f$ is bounded on $\mathcal S(N)$.  There is a Hermitian operator $T_N: N \to N$ and a bounded symmetric bilinear form $t_N$ on $N \times N$ such that $f(x)=t_N(x,x)= (T_Nx,x),$ $x \in \mathcal S(N)$.  We define a bilinear form $t$ on $S \times S$ as follows: Let $x$ and $y$ be two vectors of $S$ and let $N$ be any finite-dimensional subspace of $S,$ $\dim N \ge3,$ containing both $x$ and $y$.  Put $t(x,y)=t_N(x,y)$.  This $t$ is a well-defined bilinear form on $S \times S$ since if $M$ is another finite-dimensional subspace of $S,$ $\dim M\ge 3,$ containing $x$ and $y$, then for $Q=M\vee N$, we have $t_N(x,y)=t_Q(x,y)=t_M(x,y)$.  It is evident that $f(x)=t(x,x)$, $ x \in \mathcal S(S)$.

Then $t$ is a semibounded symmetric bilinear form on $H$ with $D(t)=S$.  The uniqueness of $t$ follows from the Gleason Theorem.
\end{proof}

We note that if $f\sim (f,S)$ is a frame type function such that $\{f(x_i)\}$ is summable for any ONB $\{x_i\}$ in $S$ and $\dim S\ne 2$, then Theorem \ref{th:2.2} can be strengthened as follows (cf. \cite[Thm 3.2.21]{DvuG}): There is a unique Hermitian trace operator $T$ on $H$ such that
$$
f(x)=(Tx,x),\quad x \in \mathcal S(S).
$$

The proof of the following important lemma by Luguvaya--Sherstnev can be found in \cite{LuSh}, see also \cite[Lem 3.4.2, Cor 3.4.3]{DvuG}. 

\begin{lemma}[Lugovaya--Sherstnev]\label{le:2.3}
Let $\dim H=3$ and let $f$ be a frame function on $H$.  Let there be two orthogonal unit vectors $x$ and $y$ such that $f(x)$ and $f(y)$ are finite.  If $W_H=\infty$ and $f(z)$ is finite, then $z \in {\rm sp}(x,y)$.

\end{lemma}

\begin{theorem}\label{th:2.4}
Let $\dim H\ne 2$ and $f$ be a semibounded frame function on $H$ such that there is an ONB $\{x_i\}$ with $f(x_i)$ finite for any unit vector $x_i$.  There is a unique semibounded symmetric bilinear form $t$ with $D(t)$ dense in $H$ such that

\[   f(x) = \left\{ \begin{array}{ll}
t(x,x) &\quad {\rm if }\
x \in \mathcal S(D(t)),\\
 \infty & \quad  {\rm if}\ x\in \mathcal S(H)\setminus \mathcal S(D(t)).
\end{array}
\right. \eqno(2.2)
\]

\end{theorem}

\begin{proof}
If $H$ is finite-dimensional, $W_H$ and $f$ are finite, the proof follows the same ideas as the first part of the proof of Theorem \ref{th:2.2} with $D(t)=H$.

Now let $\dim H=\infty$ and let $\{x_i\}$ be an ONB of $H$ such that every $f(x_i)$ is finite. Denote by $D(f):=\{x \in H\setminus\{\mathbf 0\}: f(x/\|x\|)$ is finite$\}\cup \{\mathbf 0\}$.  Then $x_i \in D(f)$ for any $x_i$.  We assert that $D(f)$ is a linear submanifold of $H$.   Indeed, let $x,y \in D(f)$ be two nonzero vectors. Without loss of generality, we can assume  $x$ and $y$ are two unit vectors from $D(f)$ that are linearly independent. If $x\perp y,$ then clearly $x+y \in {\rm sp}(x,y)$.   Now let $x\not\perp y$.

Take two vectors $x_1,x_2$ from the ONB $\{x_i\}$ and set $P =\mbox{sp}(x_1,x_2)$.  If $x,y \in P,$ then $x+y \in P$ and $x+y \in D(f)$.  Let $x \notin P$ and define $M_1 = \mbox{sp}(x_1,x_2,x)$.  We assert the weight $W_{M_1}$ of $f$ on $M_1$ is finite; if not, by the Lemma Lugovaya--Sherstnev we conclude $x \in P$ which is absurd.

If $y \in M_1,$ then clearly $x+y \in M_1$ and $x+y \in D(t)$.  Assume finally $y \notin M_1$.  Choose a unit vector $z\in M_1$ that is orthogonal to $x$.  Again by Lemma \ref{le:2.3}, we conclude that $f$ on $M_2=\mbox{sp}(x,y,z)$ has finite weight, so that $x+y \in D(f)$.

Since $D(f)$ contains the ONB $\{x_i\},$ $S:=D(f)$ is dense in $H$.  Since $f|S \sim (f|S,S)$ is a semibounded finite frame type function on $H,$ by Theorem \ref{th:2.2}, there is a unique semibounded symmetric bilinear form $t$ with dense $D(t)=S=D(f)$ such that $f(x)=t(x,x),$ $x \in \mathcal S(S)$.  Hence, (2.2) holds.
\end{proof}

It is clear that if  $t$ is a positive symmetric bilinear form with $D(t)$ dense in $H$, then $f(x):=t(x,x),$ $x \in \mathcal S(D(t)),$ is a finite frame type function on $H$.  Thus for positive  finite frame type functions and  positive symmetric bilinear forms with dense domain there is a one-to-one correspondence. It is worthy of attention that for non-positive bilinear forms such a correspondence is not true, in general, as it follows from Example \ref{ex:2.1}(iv), where there is presented a bounded bilinear form not positive that is not a frame type function.


If $t$ is a positive symmetric bilinear form with dense domain, then $f$ defined on $\mathcal S(H)$ by (2.2) is not necessarily a frame function. For example, in \cite{Lug1}, see also \cite[Thm 3.7.2]{DvuG}, it is shown that if $t$ is a singular positive symmetric bilinear form with dense domain, then $f$ defined by (2.2) is not a frame function, but $f|D(t)$ is a positive finite frame type function.

We note that according to \cite{DoSh}, \cite[Thm 3.2.20]{DvuG}, every finite frame type function $f\sim (f,S)$, such that $\{f(x_i)\}$ is summable for any ONS $\{x_i\}$ in $S$,  is bounded. We do not know whether every (our) finite frame type function is semibounded if $\dim H=\infty$. As it was already mentioned, according to \cite[Prop 3.2.4]{DvuG}, there are unbounded frame type functions for $H$, $\dim H=n \ge 2$, that are not generated by any bilinear form.

\section{Gleason Measures on $\mathcal L(H)$ and Generalized Effect Algebras}

In this section, we study finitely additive measures that are regular. Different classes of measures will be studied from the point of view of generalized effect algebras.

A finitely additive measure $m$ on $\mathcal L(H)$ is said to be (i) {\it regular} if
$$ m(M)=\sup\{m(P)\colon P\subseteq M, P\in \mathcal L(H), \dim P<\infty\},\quad M \in \mathcal L(H),
$$
and (ii) $\mathcal P(H)_1$-{\it bounded} if $\sup\{m(\mathrm{sp}(x): x \in D(m)\}< \infty$, where
$$
D(m):=\{ x\in H\colon m({\rm sp}(x))<\infty \}\cup \{\mathbf 0\}.
$$

According to \cite{DvuG}, $m$ satisfies (i) the {\it L-S property} (L-S stands for Lugovaja and Sherstnev), if there is a two-dimensional subspace $Q$ of $H$ such that $m(Q)<\infty$, (ii) the {\it density property} if
$D(m)$ is dense in $H$, and (iii) the {\it L-S density property} if both (i) and (ii) hold.

We note that the L-S property entails that $D(m)$ is a dense linear subspace of $H$.  If $o$ is the zero measure $o$, i.e. $o(M)=0$, $M\in \mathcal L(H)$, then $o$ is a regular measure with the L-S property. In addition, if $\dim H \ne 2$, then a finitely additive measure $m$ has the L-S density property iff $m$ is $\sigma$-finite.

We denote by $\mathrm{Reg}(H)$  the class of regular measures with the L-S property.

An important relationship between regular finitely additive measures with the L-S property and positive bilinear forms with  dense domain is the following result, for the proof see \cite[Lem 3.7.8. Thm 3.7.9]{DvuG}. Before that we remind an important notation: Let $t$ be a bilinear form with domain $D(t)$, let $M \in \mathcal L(H)$ be given, and let $P_M$ be the orthogonal projector of $H$ onto $M$. We define a new bilinear form $t\circ P_M$ as a bilinear form whose domain is the set $D(t\circ P_M):=\{x\in H\colon P_Mx \in D(t)\}$. If this bilinear form is determined by a trace operator $T^t_M$, we write $t\circ P_M \in \mathrm{Tr}(H)$ and $\mathrm{tr}(t\circ P_M):= \mathrm{tr}(T^t_M).$

\begin{theorem}\label{th:3.1}
Let $H$ be a Hilbert space.

{\rm (1)} Let $t$ be a
positive bilinear form such that $D(t)$ is dense in H. Then
the mapping $m_t: \mathcal L(H)\to [0,\ \infty]$ given by

$$\label{eq:3.37}
m_t(M) = \left\{ \begin{array}{ll}
{\displaystyle {\rm {\rm tr}} (t\circ P_M}) &\quad {\rm if} \ {\displaystyle  t\circ P_M \in {\rm Tr}(H)},\\[4pt]
{\displaystyle \infty} & \quad  {\rm otherwise},
\end{array}
\right. \eqno(3.1)
$$
is a regular finitely additive measure with the L-S density
property.

{\rm (2)} Let $m$ be a regular finitely
additive measure with the L-S density property on $\mathcal L(H)$, $\dim H\ne 2$. Then there
exists a unique bilinear form $t$ with domain $D(t)=D(m)$ such that {\rm (3.1)} holds.
\end{theorem}

Since formula (3.1) is a generalization of the famous Gleason formula (2.1), we call measures expressed by bilinear forms via (3.1) also Gleason measures.

We note that by Theorem \ref{th:3.1}, if $m$ is a regular finitely additive measure with the L-S property, there is a unique positive positive bilinear form $t$ with $D(t)=D(m)$ such that (3.1) holds. In other words, if $x\in D(m)$, we have $t(x,x)=\|x\|^2m(\mathrm{sp}(x))$ if $x\ne \mathbf 0$, otherwise, $t(\mathbf 0,\mathbf 0)=0.$

On the other hand, if $m$ is a finite positive regular measure on $\mathcal L(H)$, then $m$ is completely additive \cite[Thm 3.7.2]{DvuG}. For unbounded measures this is not a case, in general. More precisely, any positive singular bilinear form $t$ generates by (3.1) a regular measure that is not completely-additive, see \cite[Thm 3.7.6]{DvuG}, \cite{Lug2}. A criterion when a bilinear form $t$ generates a completely additive measure on $\mathcal L(H)$, see \cite{Lug2}, \cite[Thm 3.7.5]{DvuG}: A positive bilinear form $t$ defines through (3.1) a $\sigma$-finite completely additive measure on $\mathcal L(H)$ iff for any $M \in \mathcal L(H)$,
$$
\overline{(t\circ P_M)_r} \in \mathrm{Tr}(H) \ \mathrm{implies} \ t \circ P_M \in \mathrm{Tr}(H), \eqno(3.2)
$$
where $\overline{(t\circ P_M)_r}$ is the closure of the regular part of $t\circ P_M$.

\begin{proposition}\label{pr:3.2}
If $m_1,m_2$ are regular measures on $\mathcal L(H)$, then $m_1+m_2$ is also a regular measure.
\end{proposition}

\begin{proof}
Given $M\in \mathcal L(H)$ and $i=1,2$, there are two sequences $\{P_n^1\}$ and $\{P_n^2\}$  of finite-dimensional subspaces of $M$ such that $m_i(M)=\lim_n m_i(P_n^i)$. Then $m_i(M) = \lim_n m_i(P_n^1\vee P_n^2)$ which entails $(m_1+m_2)(M)=m_1(M)+m_2(M)=\lim_n m_1(P_n^1\vee P_n^2)+ \lim_n m_2(P_n^1\vee P_n^2)=\lim_n (m_1+m_2)(P_n^1\vee P_n^2)$ which proves regularity of $m_1+m_2$.
\end{proof}

Proposition \ref{pr:3.2} suggests to study the set of regular Gleason measures also from the point of view of generalized effect algebras.


A partial algebra $(E;\oplus ,0)$ is called a {\em generalized
effect algebra} if $0\in E$ is a distinguished element and
$\oplus $ is a partially defined binary operation on $E$ which
satisfies the following conditions for any $x,y,z\in E$:

\begin{itemize}
\setlength{\leftmargin}{1.8cm}
\item[\rm(GEi)]  $x\oplus  y=y\oplus  x$, if one side is defined,
\item[\rm(GEii)]
$(x\oplus  y)\oplus  z=x\oplus  (y\oplus  z)$, if one side is defined,
\item[\rm(GEiii)] $x\oplus  0=x$,
\item[\rm(GEiv)] $x\oplus  y=x\oplus z$ implies $y=z$
(cancellation law),
\item[\rm(GEv)]
$x\oplus  y=0$ implies $x=y=0$.
\end{itemize}

For every generalized effect algebra $E$, a partial binary
operation $\ominus $ and a relation $\leq:=\leq_\oplus$ can be defined by
\begin{itemize}
\item[\rm(ED)] $x\le y$ and $y\ominus x=z$ iff $x\oplus z$ is
defined and $x\oplus  z=y$.
\end{itemize}
Then $\le$ is a partial order on $E$ under which $0$ is the least
element of $E$. A generalized effect algebra with the top element $1 \in E$ is called an {\it effect algebra} and we usually write $(E;\oplus,0,1)$ for it.

For example, let $G$ be an Abelian po-group, i.e. a group $G=(G;+,0)$ endowed with a partial order $\le$ such that $a\le b$, then $a+c\le b+c$ for each $c \in G$. Let $G^+:=\{g\in G: g\ge 0\}$ be the positive cone of $G$. Then $(G^+;\oplus,0)$ is a generalized effect algebra, where $a\oplus b = a+b$, $a,b \in G^+$. If $u\in G^+$ is a fixed element, then the interval $[0,u]:=\{g \in G: 0\le g \le u\}$ defines an effect algebra $([0,u];\oplus,0,u)$, where $a\oplus b = a+b$, $a,b \in [0,u]$, whenever $a+b\le u.$

In particular, if $\mathcal B(H)$ is the class of all Hermitian operators of a Hilbert space $H$, then $\mathcal B(H)$ is a po-group, saying $A\le B$ iff $(Ax,x)\le (Bx,x)$, $x \in H$, and $\mathcal B(H)^+$ is a generalized effect algebra. If $\mathcal E(H):=\{A \in \mathcal B(H): O\le A \le I\}$, where $O$ and $I$ are the zero and  identity operators, respectively, then $\mathcal E(H)$ is a prototypical example of effect algebras important for Hilbert space quantum mechanics.

A subset $S$ of $E$ is called a {\it sub-generalized effect algebra} ({\it sub-effect algebra}) of $E$ iff {\rm (i)} $0 \in S$ ($1\in S$), {\rm (ii)}
if out of elements $x,y,z \in E$ such that $x\oplus y=z$ at least two are in $S$, then all $x,y,z \in S$.

The following theorem follows ideas from \cite{RZP}, where a structure of linear operators was described. 

\begin{theorem}\label{th:3.3}
Let $H$ be an infinite-dimensional complex Hilbert space. Let
$\mathrm{Reg}_f(H)$ be the set of regular finitely additive measures $m$ on $\mathcal L(H)$ with the L-S density property such that if $m$ is $\mathcal P_1(H)$-bounded, then $D(m) = H$.
Let us define a partial operation $\oplus$ on $\mathrm{Reg}_f(H)$:
For $m_1,m_2 \in \mathrm{Reg}_f(H)$, $m_1\oplus m_2$ is defined if and only if  $m_1$ or $m_2$ is $\mathcal P_1(H)$-bounded or $D(m_1) = D(m_2)$
and then $m_1\oplus m_2 := m_1+ m_2$.
Then $(\mathrm{Reg}_f(H);\oplus, o)$ is a generalized effect algebra.
\end{theorem}

\begin{proof}
Let $m_1,m_2 \in  \mathrm{Reg}_f(H)$ be measures such that $m_1 \oplus m_2$ is defined. By Theorem \ref{th:3.1}, there are unique positive positive bilinear forms $t_1, t_2$ with $D(t_1)=D(m_1), D(t_2)=D(m_2)$ such that (3.1) holds. 
Using \cite[Thm 4.1]{DvJa}, we have that $t_1 + t_2$ is a densely defined positive bilinear form on $D(t_1 + t_2) = D(t_1) \cap D(t_2)$. Let $M \in {\mathcal L}(H)$ be such that $m_1(M), m_2(M) < \infty$. Then $(m_1 \oplus m_2) (M) =  m_1(M) + m_2(M) = \tr(t_1 \circ P_{M}) + \tr(t_2 \circ P_{M}) = \tr(t_1 \circ P_M + t_2\circ P_M) = \tr((t_1 + t_2) \circ P_M)$. If $m_1(M) = \infty$ or $m_2(M) = \infty$, we have $t_1 \circ P_{M} \notin \mathrm{Tr}(H)$ or $t_2 \circ P_{M} \notin \mathrm{Tr}(H)$ which yields $(t_1 + t_2) \circ P_{M} \notin \mathrm{Tr}(H)$. Therefore, $m_1 \oplus m_2$ is determined by $t_1 + t_2$ in the sense of (3.1) i.e. it is regular and possesses the L-S density property.

(GEi) It follows from the commutativity of the usual sum $+$.

(GEii) Let  $m_1,m_2,m_3 \in  \mathrm{Reg}_f(H)$ such that $(m_1\oplus m_2)\oplus m_3$ is defined. First, let us assume that $D(m_1) = H$. Then $D(m_1)\cap D(m_2) = D(m_2)$ and $D(m_2) = D(m_3)$ or $D(m_3) = H$. That is, $m_2 \oplus m_3$ is defined with $D(m_2 + m_3) = D(m_2)$, hence also $m_1 \oplus (m_2 \oplus m_3)$ is defined. The same analogy holds for $D(m_2) = H$. If $D(m_1),D(m_2) \not= H$, we have $D(m_1) = D(m_2)$ and then $D(m_3) = D(m_2)$ or $D(m_3) = H$, which yields the existence of $m_1 \oplus (m_2 \oplus m_3)$. The equality follows from the associativity of the usual sum $+$.

(GEiii) It holds since $o \in \mathrm{Reg}_f(H)$ is $\mathcal P_1(H)$-bounded, hence $D(o) = H$.

(GEiv) Let $m_1 \oplus m_2 = m_1 \oplus m_3$.  By Theorem \ref{th:3.1}, for any $i=1,2,3$, there is a unique positive positive bilinear form $t_i$ with $D(t_i)=D(m_i)$ such that (3.1) holds.

Let $m_1, m_2, m_3$ be all $\mathcal P_1(H)$-bounded. Then $t_1(x,y) + t_2(x,y) = t_1(x,y)+t_3(x,y)$ for all $x,y \in H$, which is equivalent to $t_2(x,y)=t_3(x,y)$ for all $x,y \in H$
hence $t_2 = t_3$, so that $m_2=m_3$ by Theorem \ref{th:3.1}.

Let $m_1$ be $\mathcal P_1(H)$-bounded and let $m_2$ be not $\mathcal P_1(H)$-bounded. Then $m_3$ has to be not $\mathcal P_1(H)$-bounded and $D(m_2) = D(m_1+m_2) = D(m_1+m_3)=D(m_3)$. By the same argument as in the previous case (restricted on $D(m_2)$) we have $t_2=t_3$ and consequently,
$m_2=m_3$.

Let $m_1$ be not $\mathcal P_1(H)$-bounded and let $m_2$ be $\mathcal P_1(H)$-bounded. Then  $D(m_1) =D(m_1+m_2) = D(m_1+m_3)$. For every $x, y \in D(t_1),$ we have $t_1(x,y) + t_2(x,y) = t_1(x,y)+t_3(x,y).$ Hence, $t_2(x,y)=t_3(x,y)$, that is
$t_3{\mid D(t_1)} = t_2{\mid D(t_1)}$. Since $t_2$ is on $D(t_1)$ bounded, $t_3$ is also bounded and by \cite[Lem 4.1]{DvJa}, $t_3$ can be extended on $H$ in a unique way to a symmetric bilinear form, that is $t_2=t_3$ which implies $m_2 = m_3$.

Let $m_1$  and $m_2$ be not $\mathcal P_1(H)$-bounded. Then $D(m_1) = D(m_2) = D(m_1+m_2) = D(m_1+m_3)$ and in the same way as in the previous case, $t_3{\mid D(t_1)} = t_2{\mid D(t_1)}$. Since $m_2$ is not  $\mathcal P_1(H)$-bounded, $m_3$ is also not $\mathcal P_1(H)$-bounded. Because $m_1 \oplus m_3$ is defined, we have $D(t_1) = D(t_3) = D(t_2),$ hence $t_3=t_2$ and consequently, $m_3 = m_2$.

(GEv) Assume that there is $m_1, m_2 \in \mathrm{Reg}_f(H)$ satisfying $m_1 \oplus m_2 = o$. Then $D(m_1 \oplus m_2) = D(o) = H$ and there exist unique positive bilinear forms $t_1, t_2$ with $D(t_1)=D(t_2) = H$ such that (3.1) holds. Using \cite[Thm 4.1]{DvJa}, we have $t_1 = t_2 = o$ which means $m_1 = m_2 = o$.
\end{proof}

\begin{theorem}\label{th:3.4}
Let $H$ be a separable infinite-dimensional complex Hilbert space. Let  $\mathrm{Reg}^{\sigma}_f(H)$ be the set of $\sigma$-additive measures from $\mathrm{Reg}_f(H)$. Then $(\mathrm{Reg}^{\sigma}_f(H);$ $\oplus_{\mid\mathrm{Reg}^{\sigma}_f(H)},o)$ is a generalized effect algebra, but $\mathrm{Reg}^{\sigma}_f(H)$ is not a sub-generalized effect algebra of the generalized effect algebra $(\mathrm{Reg}_f(H);\oplus, o)$.
\end{theorem}

\begin{proof}
Clearly, $o \in \mathrm{Reg}^{\sigma}_f(H)$. It is straightforward to show that if a sum $m_1\oplus m_2$ of two measures from $\mathrm{Reg}^{\sigma}_f(H)$ exists in $\mathrm{Reg}_f(H)$, it is a measure from $\mathrm{Reg}^{\sigma}_f(H)$ as well. Let us consider a regular measure $m_1 \in \mathrm{Reg}_f(H)$ given by $m_1(M) = \dim M$, $M\in \mathcal L(H)$, which is a $P_1(H)$-bounded $\sigma$-additive measure and it corresponds to the bilinear form $t_I$ corresponding to the identity operator $I$ on $H$, i.e. $t_I(x,x)=(Ix,x)=(x,x)$, $x \in H$. Let $m_2 \in \mathrm{Reg}_f(H)$ be a finitely additive measure given by (3.1) for some positive singular bilinear form $s$ with $D(s)=H$, hence by \cite[Thm 3.7.2]{DvuG}, $m_2$ is not $\sigma$-additive. Since $D(s) = H$, it holds $s \circ P_M \in \mathrm{Tr}(H)$ if $\dim M < \infty$. Moreover, $t_I \circ P_M = P_M \in \mathrm{Tr}(H)$ iff dim $M < \infty$ which gives $(s + t_I) \circ P_M \in \mathrm{Tr}(H)$ iff $\dim M < \infty$. By \cite[Cor 2.3]{Sim}, we have $\overline{(s + t_I)_r} = t_I$ and using (3.2) or \cite[Thm 3.7.5]{DvuG}, $m_1 \oplus_{\mid\mathrm{Reg}^{\sigma}_f} m_2:=m_1\oplus m_2 =m_1 +m_2$ is a $\sigma$-additive measure. Which means that $\mathrm{Reg}^{\sigma}_f(H)$ is not a sub-generalized effect algebra of $\mathrm{Reg}_f(H)$ because $m_1, m_1 \oplus m_2 \in \mathrm{Reg}^{\sigma}_f(H)$ but $m_2 \notin \mathrm{Reg}^{\sigma}_f(H)$.
\end{proof}

\begin{remark}\label{rem:3.5}
Let $H$ be an infinite-dimensional complex Hilbert space. Let  $\mathrm{Reg}^{c}_f(H)$ be the set of completely additive measures from $\mathrm{Reg}_f(H)$.
Then by \cite[Thm 3.6.3]{DvuG} $\mathrm{Reg}^{c}_f(H) = \mathrm{Reg}^{\sigma}_f(H) $, which yields $(\mathrm{Reg}^{c}_f(H);\oplus, o)$ is a generalized effect algebra.
\end{remark}

Using Theorem \ref{th:2.2} and Theorem \ref{th:2.4}, we can extend the previous results also for frame type functions. For any two frame type function $f_1, f_2$ on $H$, their {\it usual sum} $f(x) := f_1(x) + f_2(x)$ for all $x \in {\mathcal S}(S_1) \cap {\mathcal S}(S_2)$ is again a frame type function, whenever $S_1 \cap S_2$ is dense in $H$. 

A frame function $o_f$ is defined by $o_f(x) = 0$ for all $x \in \mathcal S(H)$.

\begin{theorem}\label{th:3.6}
Let $H$ be an infinite-dimensional complex Hilbert space. Let $\mathrm{Ftf}(H)$ be the set of positive finite frame type functions $f\sim (f,S)$ on $H$ such that whenever $f$ is bounded, then $S=H$. Let us define a partial operation $\hat{\oplus}$ on $\mathrm{Ftf}(H)$:
For $f_1,f_2 \in \mathrm{Ftf}(H)$, $f_1 \hat{\oplus} f_2$ is defined if and only if  $f_1\sim (f,S_1)$ or $f_2\sim (f,S_2)$ is bounded or $S_1 = S_2$
and then $f_1\hat{\oplus} f_2 := f_1+ f_2$. Then $(\mathrm{Ftf}(H);\hat{\oplus}, o_f)$ is a generalized effect algebra.
\end{theorem}

\begin{proof}
By Theorem \ref{th:2.2}, there is a one-to-one correspondence between positive finite frame type functions and positive bilinear forms given by $f(x) = t(x,x)$ for $x \in D(t)$ and $S = D(t)$. A positive frame type function $f$ is bounded iff the corresponding bilinear form $t$ is bounded iff the corresponding finitely additive measure $m$ is bounded, and clearly $S = D(t) = D(m)$. Hence, the sets $\mathrm{Ftf}(H)$ and $\mathrm{Reg}_f(H)$ from Theorem $\ref{th:3.4}$ are in a one-to-one correspondence.

For any two frame type functions $f_1, f_2 \in \mathrm{Ftf}(H)$, $f_1 \hat{\oplus} f_2$ is defined if and only if for finitely additive measures $m_1, m_2$, given by bilinear forms $t_1(x,x) := f_1(x),t_2(x,x):=f_2(x)$, $m_1 \oplus m_2$ is defined and then $(f_1 \hat{\oplus} f_2)(x) = (m_1 \oplus m_2)({\rm sp}(x))$. That is, the proof follows the same arguments as the proof of Theorem \ref{th:3.3}.
\end{proof}

\begin{remark}
Let $f\sim (f,S)$ be a frame type function on $H$. If the function $\overline{f}$ given by $\overline{f}(x):= f(x)$ for all $x \in S$ and $f(x):= \infty$ otherwise is a frame function on $H$, then we say that $f$ {\it induces} a frame function $\overline{f}$.
\end{remark}

\begin{theorem}\label{th:3.8}
Let $H$ be an infinite-dimensional complex Hilbert space. Let $\mathrm{Ff}(H) \subseteq \mathrm{Ftf}(H)$ be the set of all frame type functions $f$ on $H$ which induce a frame function $\overline{f}$. Then $\mathrm{Ff}(H)$ is a subset of $\mathrm{Ftf}(H)$ but it is not a sub-generalized effect algebra of $(\mathrm{Ftf}(H);\hat{\oplus}, o_f)$, but $(\mathrm{Ftf}(H);\hat{\oplus}_{\mid\mathrm{Ftf}(H)},o_f)$ is a generalized effect algebra on its own.
\end{theorem}

\begin{proof}
A restriction of any finite frame function $\overline{f}$ on $S = \{x \in H \mid f(x) < \infty \}$ is a finite frame type function $f\sim(f, S)$, and on the other hand, $f$ induces $\overline{f}$. Since there is a one-to-one correspondence between the set of positive frame functions $\mathrm{Ff}(H)$ and the set of completely additive measures $\mathrm{Reg}^{c}_f(H)$, the theorem follows from Theorem \ref{th:3.4} and Remark \ref{rem:3.5}.
\end{proof}

\section{Monotone convergence properties of $\sigma$-additive measures}

In the section, we will study some monotone Dedekind upwards (downwards) $\sigma$-complete properties of generalized effect algebras of regular measures.

In \cite[Thm 3.10.1]{DvuG}, there was proved an analogue of the Nikod\'ym convergence theorem which says following:

\begin{theorem}\label{th:3.10.1}
Let $\{m_n\}$ be a sequence of finite signed $\sigma$-additive measures on $\mathcal L(H)$ of a Hilbert space $H$. If there is a finite limit $m(M) = \lim_{n \to \infty} m_n (M)$ for any $M \in \mathcal L(H)$, then $m$ is a finite signed measure on $\mathcal L(H)$, and $\{m_n\}$ is uniformly $\sigma$-additive with respect to $n$.
\end{theorem}

We say that a generalized effect algebra $E$ is (i) {\it monotone Dedekind upwards $\sigma$-complete} if, for any sequence $x_1 \le x_2\le \cdots,$ which is  dominated  by some element $x_0$, i.e. each $x_n \le x_0,$ the element $x = \bigvee_n x_n$  is defined in $E$ (we write $\{x_n\}\nearrow x$), (ii) {\it monotone Dedekind downwards $\sigma$-complete} if, for any sequence $x_1 \ge x_2\ge \cdots,$  the element $x = \bigwedge_n x_n$  is defined in $E$ (we write $\{x_n\}\searrow x$). If $E$ is an effect algebra,  both later  notions are equivalent. In addition, we say that a generalized effect algebra $E$ is {\it upwards directed} if, given $a_1,a_2 \in E,$ there is $a \in E$ such that $a_1,a_2 \le a.$

\begin{theorem}\label{mdd}
Let ${H}$ be an infinite-dimensional complex Hilbert space. Then the generalized effect algebra $(\mathrm{Reg}_f(H);\oplus, o)$ is
monotone downwards $\sigma$-complete, but it is not Dedekind monotone upwards $\sigma$-complete.
\end{theorem}

\begin{proof}
This holds since $(\mathrm{Reg}_f(H);\oplus, o)$ is isomorphic to the generalized effect algebra of all positive bilinear forms on $H$, which by \cite[Thm 5.2]{DvJa} has the mentioned properties.
\end{proof}

\begin{theorem}
Let $H$ be an infinite-dimensional complex Hilbert space. The generalized effect algebra $(\mathrm{Reg}^{\sigma}_f(H);\oplus, o)$ is neither
monotone downwards $\sigma$-complete, nor Dedekind monotone upwards $\sigma$-complete.
\end{theorem}

\begin{proof}
Let us consider the complex Hilbert space $H = L^2(0,1)$ and the following sequence of bilinear forms $\{s_n\}_{n\in \mathbb{N}}$ given by
$$
s_n(u,u):= (1+ \frac{1}{n}) \int_0^1 \left|u'(x)\right|^2 {\rm d}x + \left|u(0)\right|^2 + \left|u(1)\right|^2,\eqno(4.1)
$$
where $u'$ is derivative of $u \in L^2[0,1]$. Let
$$s(u,u):= \int_0^1 \left|u'(x)\right|^2 {\rm d}x + \left|u(0)\right|^2 + \left|u(1)\right|^2 ,$$
$$\hat{s}(u,u) := \int_0^1 \left|u'(x)\right|^2 {\rm d}x, $$
$$ s_0(u,u):= \left|u(0)\right|^2 + \left|u(1)\right|^2. $$

Then $s, \hat{s}$ and $s_n$ are closed bilinear forms for all $n \in \mathbb{N}$. According to \cite[Prop 3.7.4]{Dvu1}, measures $m_s, m_{\hat{s}}$ and $m_{s_n}$ induced by these bilinear forms via (3.1) are $\sigma$-additive (due to their regularity also completely additive). Also $s-s_n$, $\hat{s}-s_n$ are closed bilinear forms, which means that $m_s, m_{\hat{s}} \leq_{\sigma} m_{s_n}$ for all $n\in \mathbb{N}$, where $\leq_\sigma$ is the partial order in the generalized effect algebra $\mathrm{Reg}^\sigma_f(H)$ induced by the partial addition in it via (ED). It holds $m_{\hat{s}} \leq_\sigma m_s$ in the generalized effect algebra $\mathrm{Reg}_f(H)$, but since $s-\hat{s}$ is a singular bilinear form, the corresponding measure $m_{s-\hat{s}}$ is not $\sigma$-additive hence $m_{\hat{s}} \nleq_{\sigma} m_s$. Since $m_s = \bigvee_{n\in \mathbb{N}} m_{s_n}$ in the generalized effect algebra $\mathrm{Reg}_f(H)$, the sequence $\{m_{s_n}\}_{n\in \mathbb{N}}$ has no infimum in $\mathrm{Reg}^{\sigma}_f(H)$. By \cite[Lemma 5.1]{DvJa}, $\mathrm{Reg}^{\sigma}_f(H)$ is not a Dedekind upwards $\sigma$-complete generalized effect algebra.
\end{proof}

\end{document}